%% file: main.tex
\documentclass[conference]{IEEEtran}

\makeatother
\usepackage{amssymb,amsmath}
\usepackage{bbding}
\usepackage{flushend,cuted}
\usepackage{amsthm}
\usepackage{amsfonts}
\usepackage{float} 
\usepackage{srcltx}
\usepackage{mathrsfs} 
\usepackage{multirow}
\usepackage{amssymb}

\usepackage{graphicx}
\usepackage{epsfig}
\usepackage{psfrag}
\usepackage{subfigure}

\usepackage{setspace}
\usepackage{multicol}
\usepackage[noadjust]{cite}
\usepackage{hyperref}

\usepackage[]{units} 
\usepackage{url} 
\usepackage[dvips]{color}
\usepackage{verbatim} 
\usepackage{cite} 
\usepackage{ifthen} 
\usepackage{ifpdf} 
\usepackage{soul}
\usepackage{mathtools}

\newtheorem{lemma}{Lemma}

\newtheorem{theorem}{Theorem}
\newtheorem{definition}{Definition}

\makeatletter
\def\widebreve#1{\mathop{\vbox{\m@th\ialign{##\crcr\noalign{\kern3\p@}%
      \brevefill\crcr\noalign{\kern3\p@\nointerlineskip}%
      $\hfil\displaystyle{#1}\hfil$\crcr}}}\limits}

\def\brevefill{$\m@th \setbox\z@\hbox{$\braceld$}%
  \bracelu\leaders\vrule \@height\ht\z@ \@depth\z@\hfill\braceru$}
\renewcommand*\env@matrix[1][*\c@MaxMatrixCols c]{%
  \hskip -\arraycolsep
  \let\@ifnextchar\new@ifnextchar
  \array{#1}}
\makeatletter


\input{rv_defs.tex}

\newcommand{\SNR}{\text{$\mathsf{SNR}$}}
\DeclareRobustCommand{\prob}[1][{\rm Pr}]{\ensuremath {{#1}}}
\DeclareRobustCommand{\aNorm}[1][aNorm]{\ensuremath {\|{\bf a}\|}}
\DeclareRobustCommand{\Nt}[1][Nt]{\ensuremath {N_t}}
\DeclareRobustCommand{\alNt}[1][Nt]{\alpha(\Nt)}
\DeclareRobustCommand{\genGam}[1][\beta]{\ensuremath {{#1}}}

\usepackage{lipsum}

\usepackage{fancyhdr}

\begin{document}
\allowdisplaybreaks

\title{Explicit Lower Bounds on the Outage Probability of Integer Forcing over $N_r \times 2$ Channels}



\author{
\IEEEauthorblockN{Elad Domanovitz and Uri Erez}
\IEEEauthorblockA{
Dept. EE-Systems \\
Tel Aviv University, Israel
}

}
\maketitle

\begin{abstract}
%
The performance of integer-forcing equalization for communication over the compound multiple-input multiple-output channel is investigated.
An upper bound on the resulting outage probability as a function of the gap to capacity has been derived previously, assuming a random precoding
matrix drawn from the circular unitary ensemble is applied prior to transmission. In the present work a simple and explicit lower bound on the worst-case outage probability is derived for the case of a system with two transmit antennas and two or more receive antennas, leveraging the properties of the Jacobi ensemble. The derived lower bound is also extended to random space-time precoding,  and may serve as a useful benchmark for assessing the relative merits of various algebraic space-time precoding schemes.
We further show that the lower bound may be adapted to the case of a $1 \times N_t$ system. As an application of this, we derive closed-form bounds for the symmetric-rate capacity of the Rayleigh fading multiple-access channel where all terminals are equipped with a single antenna. Lastly, we  demonstrate that the integer-forcing equalization coupled with distributed space-time coding is able to approach these bounds.
\end{abstract}


%
\section{Introduction}
%
\label{sec:intro}
This paper addresses communication over a compound multiple-input multiple output (MIMO) channel, where the transmitter only knows the number of transmit antennas and the mutual information.
More specifically, the goal of this work is to
assess the performance
of (randomly precoded) integer-forcing (IF) equalization for such a scenario. 

Communication over the compound MIMO channel using an architecture employing space-time linear processing at the transmitter side and IF equalization at the receiver side was proposed in \cite{OrdentlichErez:IFUniversallyAchievesCapacityUpToGap:2013}.
It was shown that such an architecture \emph{universally} achieves capacity up to a constant gap, provided that the precoding matrix corresponds to a linear perfect space-time code \cite{elia2007perfect}, \cite{oggier2006perfect}.

Recently, in \cite{domanovitz2017outage}, the outage probability of IF where  random unitary precoding is applied over the spatial dimension only was considered and
an explicit universal upper bound on the outage probability for a given target rate and gap to capacity was derived.

In the present work we derive an explicit lower bound on this outage probability for the case of a system with two transmit antennas. We further
extend the framework of \cite{DomanovitzE16} by considering also
space-time random unitary precoding  (rather than space-only).

\section{Problem Formulation and Preliminaries}
\label{sec:channel_model}
\subsection{Channel Model}
The  (complex) MIMO channel is described
by\footnote{We denote all complex variables with $c$ to distinguish them from their real-valued representation.}
\begin{align}
\boldsymbol{y}_c=\svv{H}_c\boldsymbol{x}_c+\boldsymbol{z}_c,
\label{eq:channel_model}
\end{align}
where $\svv{x}_c\in \mathbb{C}^{\Nt}$ is the channel input vector, $\svv{y}_c \in \mathbb{C}^{N_r}$ is the channel output vector, $\svv{H}_c$ is an $N_r\times \Nt$ complex channel matrix, and
$\boldsymbol{z}_c$ is an additive noise vector of i.i.d. unit-variance circularly symmetric complex Gaussian random variables.
We assume that the channel is fixed throughout the transmission period.
Further, we may assume without loss of generality that the input vector $\boldsymbol{x}_c$ is subject to the power constraint\footnote{We denote by $[\cdot]^T$, the transpose of a vector/matrix and by $[\cdot]^H$, the Hermitian transpose.}
\begin{eqnarray}
\mathbb{E}(\boldsymbol{x}_c^H\boldsymbol{x}_c)\leq \Nt.
\nonumber
\end{eqnarray}

Consider the mutual information achievable with a Gaussian isotropic or ``white'' input (WI)
\begin{align}
C&=\log \det \left(\svv{I}_{N_r}+ \svv{H}_c\svv{H}_c^H\right) \nonumber \\
&=\log \det \left(\svv{I}_{N_t}+ \svv{H}_c^H\svv{H}_c\right)
\label{eq:c_WI}
\end{align}
We may define the set of all channels with $N_t$ transmit antennas
(and arbitrary $N_r$) having  the same WI mutual information $C$
\begin{align}
\mathbb{H}(C) = \left\{ \svv{H}_c : \log \det \left(\svv{I}_{N_t} +  \svv{H}_c \svv{H}^H_c\right) = C\right\}.
\label{eq:setOfChannels}
\end{align}
The corresponding compound channel model is defined by (\ref{eq:channel_model}) with the channel matrix $\svv{H}_c$ arbitrarily chosen from the set $\mathbb{H}(C)$. The matrix $\svv{H}_{c}$ is known to the receiver, but not to the transmitter. Clearly, the capacity of this compound channel is $C$, and is achieved with an isotropic Gaussian input.

Applying  the singular-value decomposition (SVD) to the channel matrix, $\svv{H}_c=\svv{U}_c\svv{\Sigma}_c\svv{V}_c^H$, we note that the unitary matrices have no impact on the mutual information.  Let $\svv{D}_c$ be defined by
\begin{align}
    \left(\svv{I}_{N_t} +  \svv{H}_c^H \svv{H}_c\right)=\svv{U}_c\svv{D}_c\svv{U}_c^H.
\end{align}
and note that $\svv{D}_c=\svv{I}+\svv{\Sigma}_c^H\svv{\Sigma}_c$.
Thus, the compound set (\ref{eq:setOfChannels}) may equivalently be
described by constraining $\svv{D}_c$ to belong to the set
\begin{align}
\mathbb{D}(C) = \left\{N_t \times N_t~ {\rm diagonal}~\svv{D}_c : \sum \log(d_{i,i}) = C\right\}.
\label{eq:dc}
\end{align}





We turn now to the performance of IF.
It has been observed that employing the IF receiver allows approaching $C$ for ``most'' but not all matrices $\svv{H}_c\in\mathbb{H}(C)$.
In the present work, we quantify the  measure of the set of bad channel matrices by considering outage events, i.e., those events (channels)
where integer forcing fails even though the channel has sufficient mutual information.
The probability space here is induced by considering
a randomized scheme where a random unitary
precoding matrix $\svv{P}_c$ is applied prior
to transmission over the channel.

More specifically, denoting by $R_{\rm IF}(\svv{H}_c)$ the rate achievable with IF over a channel $\svv{H}_c$, the achievable rate of the randomized scheme is $R_{\rm IF}(\svv{H}_c\cdot\svv{P}_c)$. As $\svv{P}_c$ is drawn at random, the latter rate is also random.
Following  \cite{DomanovitzE16}, we define the worst-case (WC) outage probability of randomized IF as
\begin{align}
     P^{\rm WC,IF}_{\rm out}(C,\Delta C)= \sup_{\svv{H}_c\in\mathbb{H}(C_{})}\prob\left(R_{\rm IF }(\svv{H}_c\cdot\svv{P}_c)<C-\Delta C \right),
     \label{eq:P_WC_OUT}
\end{align}
where the probability is with respect to the ensemble of precoding matrices and $R_{\rm IF}(\cdot)$ is the achievable rate of IF as given in \cite{IntegerForcing}.


Note that in (\ref{eq:P_WC_OUT}), we take the supremum over the entire compound class rather than taking the average with respect to some putative distribution over $\mathbb{H}(C_{})$. It follows that $P^{\rm WC,IF}_{\rm out}(C,R)$ provides an upper bound on the outage probability that holds for any such distribution.

Clearly (\ref{eq:P_WC_OUT}) is not an explicit bound. Nonetheless, by restricting attention to a uniform (Haar) measure over the unitary precoding matrices, we are  able
to obtain closed-form upper as well as lower bounds.

Specifically, we consider precoding matrices $\svv{P}_c$ drawn from the circular unitary ensemble (CUE), see e.g., \cite{metha1967random}.
Applying  the SVD to the effective channel, we have
    $\svv{H}_c\svv{P}_c=\svv{U}_c\svv{\Sigma}_c\svv{V}_c^T\svv{P}_c$.
From the properties of the CUE it follows that $\svv{V}_c^T\svv{P}_c$ has the same
(CUE) distribution as $\svv{P}_c$.
Thus, $\svv{V}_c$ (and of course  $\svv{U}_c$) plays no role in (\ref{eq:P_WC_OUT}) and we may rewrite the latter
as
\begin{align}
     P^{\rm WC,IF}_{\rm out}(C,\Delta C)= \sup_{\svv{D}_c\in\mathbb{D}(C_{})}\prob\left(R_{\rm IF }(\svv{D}_c\cdot\svv{P}_c)<C-\Delta C \right),
     \label{eq:P_WC_OUT2}
\end{align}
and thus the analysis for CUE precoding is greatly simplified.\footnote{We note that in many natural statistical scenarios, including that of an i.i.d. Rayleigh fading environment, the random transformation is actually performed by nature.}

Both the upper and lower bounds for (\ref{eq:P_WC_OUT2}) developed below heavily rely on the well-studied properties of the CUE. For the lower bound we utilize, following the approach of \cite{dar2013jacobi}, the Jacobi distribution \cite{edelman2005random} which gives the eigenvalue distribution of submatrices of such matrices.

We similarly denote by $P^{\rm WC,IF-SIC}_{\rm out}(C,\Delta C)$ the WC outage probability of IF with successive interference cancellation (SIC), the rate of which we denote by  $R_{\rm IF-SIC}(\svv{H}_c)$ and for which we give an explicit expression next.


\subsection{Integer-Forcing Equalization: Achievable Rates}
\label{sec:IFequalizationSubSecNoSIC}

We begin by recalling the achievable rates  of the IF equalization scheme, where the reader is referred to \cite{IntegerForcing} and \cite{OrdentlichErezNazer:2013} for the derivation, details and proofs. Furthermore, we follow
the notation of these works, and in particular we present IF over the reals. We also focus our attention
on IF receivers employing successive interference cancellation (SIC).

For a given choice of (invertible) integer matrix $\svv{A}$, let $\svv{L}$ be defined by the following Cholesky decomposition
\begin{align}
    \label{eq:Kzz}
    \svv{A}\left(\svv{I}+\svv{H}^T\svv{H}\right)^{-1}\svv{A}^T=\svv{L}\svv{L}^T.
\end{align}
Denoting by $\ell_{m,m}$ the diagonal entries of $\svv{L}$, IF-SIC can achieve  \cite{OrdentlichErezNazer:2013} any rate satisfying
$R<R_{\rm IF-SIC}(\svv{H})$ where
\begin{align}
 R_{\rm IF-SIC}(\svv{H})&=2\Nt\frac{1}{2}\max_{\svv{A}}\min_{m=1,...,2\Nt}\log\left(\frac{1}{\ell_{m,m}^2}\right) 
 \label{eq:IF-SIC-rate}
\end{align}
and the maximization is over all  $2\Nt\times 2\Nt$ full-rank integer matrices.



\subsection{The Jacobi Ensemble}


In the analysis we carry out, the distribution of the singluar values of a \emph{submatrix} of $\svv{P}_c$ will play a central role. To that end, we recall the Jacobi ensemble which is defined as follows
\begin{definition}
(Jacobi ensemble).  The  $\mathcal{J}(m_1, m_2, n)$ ensemble, where $m_1, m_2 \geq n$, is an $n \times  n$  Hermitian
matrix which can be constructed as $\svv{A}\left(\svv{A}+\svv{B}\right)^{-1}$, where $\svv{A}$ and $\svv{B}$ belong to the Wishart ensembles $\mathcal{W}(m_1, n)$ and $\mathcal{W}(m_2, n)$, respectively.
\end{definition}

We recall the well-known (see \cite{dar2013jacobi} and references therein)
joint probability density function of the ordered eigenvalues \mbox{$0\leq\lambda_1\leq\cdots\lambda_n\leq 1$} of the Jacobi ensemble $\mathcal{J}(m_1,m_2,n)$. Namely,
\begin{align}
&f(\lambda_1,\cdots,\lambda_n)=\nonumber \\
&\kappa^{-1}(m_1,m_2,n)\prod_{i=1}^{n}\lambda_i^{m_1-n}(1-\lambda_i)^{m_2-n}\prod_{i<j}(\lambda_i-\lambda_j)^2,
\label{eq:jacobi}
\end{align}
where $\kappa^{-1}(m_1,m_2,n)$ is a normalizing factor (Selberg integral), i.e. (see, e.g., \cite{selberg1944remarks}),
\begin{align}
\kappa_{m_1,m_2,n}=\prod_{j=1}^{n}\frac{\Gamma(m_1-n+j)\Gamma(m_2-n+j)\Gamma(1+j)}{\Gamma(2)\Gamma(m_1+m_2-n+j)}.
\end{align}

As detailed in \cite{dar2013jacobi}, the singular values of the $\frac{\Nt}{2}\times k$ submatrix of the $\Nt \times \Nt$ unitary  matrix $\svv{P}_c$ have the following Jacobi distribution:
\begin{itemize}
\item
     When $1\leq k \leq \frac{\Nt}{2}$, the \emph{singular} values of the submatrix have the same distribution as the \emph{eigenvalues} of the Jacobi ensemble $\mathcal{J}(\frac{\Nt}{2},\frac{\Nt}{2},k)$.
\item
    When $\frac{\Nt}{2} < k \leq \Nt$, using Lemma~1 in \cite{dar2013jacobi}, we have that the \emph{singular} values of the submatrix have the same distribution as the \emph{eigenvalues} of the Jacobi ensemble $\mathcal{J}(\frac{\Nt}{2},\frac{\Nt}{2},\Nt-k)$.
\end{itemize}

\section{Closed-Form Bounds for $N_r\times2$ channels}
\subsection{Space-Only Precoded Integer-Forcing}
\subsubsection{Upper Bound}
\label{sec:precIFequalizationWithSic}

We recall known upper bounds for the achievable WC outage probability of CUE-precoded IF-SIC for $N_r\times 2$ channels. The following theorem  combines Theorem 2, Lemma 4 and Corollary 2 of \cite{DomanovitzE16}.
\begin{theorem}
\label{thm:thm2}{\cite{DomanovitzE16}}
For any $N_r\times 2$ complex channel $\svv{H}_c$ with white-input mutual information $C>1$, i.e., $\svv{D}\in\mathbb{D}(C)$, and for $\svv{P}_c$ drawn from the CUE (which induces a real-valued precoding matrix $\svv{P}$), we have
\begin{align}
P^{\rm WC,IF-SIC}_{\rm out}\left(C,\Delta C\right)\leq 81\pi^2 2^{-\Delta C },
\end{align}
for $\Delta C>1$.
A tighter yet less explicit bound is
\begin{align}
P^{\rm WC,IF-SIC}_{\rm out}\left(C,\Delta C\right)\leq \max_{d_{\max}}
\sum_{{\bf a}\in\mathbb{B}(\genGam,d_{\max})}\frac{2\pi^2{2^{\nicefrac{-3}{4}(C+\Delta C)}}}{{\pi^2}\frac{\aNorm^3}{2^C}\sqrt{d_{\max}}},
\nonumber
\end{align}
where $\displaystyle{d_{\max}=\max_{i}d_i}$ and
{\small
\begin{align}
\mathbb{B}(\genGam,d_{\max})=\left\{{\bf a}:0<\|{\bf a}\|<\sqrt{{\genGam}d_{\max}} \:  {\rm and} \: \nexists 0<c<1 \: : \: c{\bf a}\in\mathbb{Z}^n \right\} \nonumber
\end{align}
}%
with $\genGam=2^{\nicefrac{-1}{2}(C+\Delta C)}$.
\end{theorem}

\subsubsection{Lower Bound on the Outage Probability via Maximum-Likelihood Decoding}
\label{sec:lowerbound}
It is natural to compare the performance attained by an IF receiver with that of an optimal maximum likelihood (ML) decoder for the same precoding scheme but where each stream
is coded using an independent Gaussian codebook. 
Since we are confining the encoders to operate in
parallel (independent streams), we are in fact considering
coding over a MIMO multiple-access channel (MAC).

Thus, a simple upper bound on the achievable rate
of integer-forcing  is the capacity of the MIMO MAC with
independent Gaussian codebooks of equal rates \cite{IntegerForcing}. Specifically,
let $\svv{H}_S$ denote the submatrix of $\svv{H}_c\svv{P}_c$ formed by taking the columns with indices in $S\subseteq \{1, 2,..., \Nt \}$. For a joint ML decoder, the maximal achievable rate  over the considered MIMO multiple-access channel is
\begin{align}
    R_{\rm ML}=\min_{S\subseteq\{1, 2,\ldots,\Nt\}}\frac{\Nt}{|S|}\log\det\left(\svv{I}_{N_r}+\svv{H}_S\svv{H}_S^H\right).
    \label{eq:R_ML_spaceOnly}
\end{align}
Note that since  $\svv{H}_c\svv{P}_c$  and thus also $\svv{H}_S$ depends on the random
precoding matrix $\svv{P}_c$, $ R_{\rm ML}$ is a random variable.

We next derive the exact WC scheme outage for ML decoding when CUE precoding is applied (with independent Gaussian codebooks) over a MIMO channel with two transmit antennas.

When $\Nt=2$, the SVD decomposition of $\svv{H}_c\svv{P}_c$ can be written as
\begin{align}
     \svv{H}_c\svv{P}_c=\svv{U}_c
     \begin{bmatrix}
        \sqrt{\rho_1} & 0 & 0 & \cdots & 0 \\
        0 & \sqrt{\rho_2} & 0 & \cdots & 0
     \end{bmatrix}^H
     \svv{V}_c^H\svv{P}_c,
    \label{eq:SVDofHc}
\end{align}
where $\rho_i=\svv{\Sigma}_{i.i}^2$. Substituting the latter in (\ref{eq:c_WI}) yields
\begin{align}
C_{} =\log(1+\rho_1)+\log(1+\rho_2).
\label{eq:C_rho1_rho2}
\end{align}

\begin{theorem}

For a CUE-precoded $N_r\times 2$ compound MIMO channel with white-input mutual information $C$ and $N_r\geq 2$, we have
\begin{align}
    P^{\rm WC}_{\rm out,ML}(C,\Delta C) = 1-\sqrt{1-2^{-\Delta C}}.
\end{align}
\label{thm:thm1}
\end{theorem}
\begin{proof}
The capacity (\ref{eq:R_ML_spaceOnly}) of the $N_r \times 2$  MIMO MAC channel with equal user rates is given by
\begin{align}
 R_{\rm ML}(\svv{H}_c\svv{P}_c)&=\min_k \min_{S\in \mathcal{S}_k}\frac{2}{k} \log\det\left(\svv{I}_{N_r}+\svv{H}_{S}\svv{H}_{S}^H\right) \nonumber \\
 &\triangleq \min_k\min_{S\in \mathcal{S}_k} C(S)
 \label{eq:joint,st}
\end{align}
where $\mathcal{S}_k$ is the set of all the subsets of cardinality $k$ from $\{1,2\}$. Hence $\svv{H}_{S}$ is a submatrix of $\svv{H}_c\svv{P}_{c}$ formed by taking $k$ columns ($k$ equals $1$ or $2$). Since we assume that $\svv{P}_c$ is drawn from the CUE, it follows that $\svv{P}_c$ is equal in distribution to $\svv{V}_c^H\svv{P}_c$. Hence, taking $k$ columns from $\svv{H}_c\svv{P}_c$ is equivalent to multiplying $\svv{H}_c$ with $k$ columns of $\svv{P}_c$. Therefore (\ref{eq:joint,st}) can be written as
\begin{align}
    R_{\rm ML}(\svv{H}_c\svv{P}_c)&=
    \min\left\{C(\{1\}),C(\{2\}),C(\{1,2\})\right\}. 
    \label{eq:Rj_k123}
\end{align}
When $k=2$, we have $C\left(\{1,2\}\right)=C_{}$.
Plugging this into (\ref{eq:Rj_k123}), we get
\begin{align}
R_{\rm ML}(\svv{H}_c\svv{P}_c)=\min\left\{C(\{1\}),C(\{2\}),C_{}\right\}.
\label{eq:Rj1111}
\end{align}

We now turn to study  $C(\{1\})$.
Note that
\begin{align}
\log\det\left(\svv{I}_{N_r}+\svv{H}_{S}\svv{H}_{S}^H\right)=\log\left(1+\svv{H}_{S}^H\svv{H}_{S}\right),
\end{align}
so that
\begin{align}
    C(\{1\})&=2\log\left(1+\begin{bmatrix}\svv{P}_{1,1}\\\svv{P}_{1,2}\end{bmatrix}^H\begin{bmatrix}
 {\rho_1} & \svv{0} \\ \svv{0} & {\rho_2}
 \end{bmatrix}\begin{bmatrix}\svv{P}_{1,1}\\\svv{P}_{1,2}\end{bmatrix}\right) \nonumber \\
 &=2\log\left(1+\rho_1\svv{P}_{1,1}^H\svv{P}_{1,1}+\rho_2\svv{P}_{2,1}^H\svv{P}_{2,1}\right).
\end{align}
Also, since $\svv{P}_{1,1}$ and $\svv{P}_{2,1}$ form a vector in a unitary matrix,
\begin{align}
\svv{P}_{1,1}^H\svv{P}_{1,1}+\svv{P}_{2,1}^H\svv{P}_{2,1}=1,
\label{eq:24}
\end{align}
and hence
\begin{align}
R\left(\{1\}\right)&=2\log\left(1+\rho_1|\svv{P}_{1,1}|^2+\rho_2(1-|\svv{P}_{1,1}|^2\right) \nonumber \\
&=2\log\left(1+\rho_2+|\svv{P}_{1,1}|^2(\rho_1-\rho_2)\right).
\end{align}

Without loss of generality we assume that $\rho_2\leq\rho_1$.  Therefore,
\begin{align}
&\prob\left(C(\{1\})<R\right)=\nonumber \\
&\prob\left(2\log\left(1+\rho_2+|\svv{P}_{1,1}|^2(\rho_1-\rho_2)\right)<R\right)= \nonumber \\
&\prob\left(|\svv{P}_{1,1}|^2<\frac{2^{R/2}-1-\rho_2}{\rho_1-\rho_2}\right),
\end{align}
where $0\leq\rho_2\leq 2^{C/2}-1$.

The probability density function of the squared
magnitude of any entry of an $M\times M$ matrix drawn from the circular unitary ensemble is \cite{NarulaTrottWornell:1999}:\footnote{It is readily seen that this distribution is a special case of the Jacobi distribution.}
\begin{align}
    f_{|\svv{P}_{1,1}|^2}(\mu)
    \begin{cases}
    (M-1)(1-\mu)^{M-2} & 0\leq\mu\leq 1 \\
    0 & \text{otherwise}
    \end{cases},
\end{align}
where the expression holds for $M\geq2$. In our case, $M=2$, and thus
$|\svv{P}_{1,1}|^2\sim U(0,1)$.
Hence,
\begin{align}
\prob\left(C(\{1\})<R\right)&=\max\left(\frac{2^{R/2}-1-\rho_2}{\rho_1-\rho_2},0\right) \nonumber \\
&\geq \frac{2^{R/2}-1-\rho_2}{\rho_1-\rho_2}.
\end{align}
As from (\ref{eq:C_rho1_rho2}) we have $\rho_1=\frac{2^C_{}}{1+\rho_2}-1$,
it follows that
\begin{align}
\prob\left(C(\{1\})<R)\right)=\frac{2^{R/2}-1-\rho_2}{\frac{2^C_{}}{1+\rho_2}-1-\rho_2}.
\label{eq:Rj_k1}
\end{align}
Now, by symmetry, it is clear that
\begin{align}
\prob\left(C(\{2\})<R)\right)=\prob\left(C(\{1\})<R)\right).
\label{eq:Rj_k1_b}
\end{align}
Furthermore, it is not difficult to show (a proof appears in Appendix~\ref{sec:appb})
that the events $\mbox{\{C(\{1\})<R\}}$
and $\mbox{\{C(\{2\})<R\}}$ are disjoint. Due to this and  by (\ref{eq:Rj_k1}) and (\ref{eq:Rj1111}), it follows that\footnote{For the case of $N_r=1$, the exact outage probability is given by (\ref{eq:new}),  setting $\rho_2=0$. \label{footnote5}}
\begin{align}
    \prob\bigl(R_{\rm ML}(\svv{H}_c\svv{P}_c)<R\bigr)&=2 \cdot
    \prob\left(C(\{1\})<R\right) \nonumber \\
    &= 2 \cdot \frac{2^{R/2}-1-\rho_2}{\frac{2^C_{}}{1+\rho_2}-1-\rho_2},
    \label{eq:new}
\end{align}
which implies that
\begin{align}
    P^{\rm WC}_{\rm out,ML}(C_{},R) = \max_{0\leq\rho_2\leq2^{C/2}-1} 2\cdot\frac{2^{R/2}-1-\rho_2}{\frac{2^C_{}}{1+\rho_2}-1-\rho_2}.
\end{align}

It is readily verified that the derivative of the expression that is maximized with respect to $\rho_2$ is zero
for (and only for)
\[
\rho^*_2=2^{-R/2-1}\left(2^{C+1}-2^{R/2+1}-2\sqrt{2^{2C}-2^{C+R}}\right),\]
and moreover, that the second derivative at this point is negative, and  hence this is a global maximum.
Finally, by plugging $\rho_2=\rho^*_2$ (and noting that $R=C-\Delta C$), we obtain
\begin{align}
    P^{\rm WC}_{\rm out,ML}(C_{},\Delta C) = 1-\sqrt{1-2^{-\Delta C}}.
\end{align}
\end{proof}

\subsubsection{Comparison of Bounds and Empirical Results}
\label{sec:comparison}
Figure~\ref{fig:xxx} depicts the lower and upper bounds as well as results of an empirical simulation of the scheme.\footnote{Rather than plotting the WC outage probability, we plot its complement.}
We observe that for \mbox{$N_r \times 2$} channels, the empirical performance of randomly precoded IF-SIC is very close to the upper (ML) bound. This suggests that one can expect that the ML bound may serve as a useful design tool for more general cases
($N_t>2$).
\begin{figure}
\begin{center}
\includegraphics[width=0.9\columnwidth]{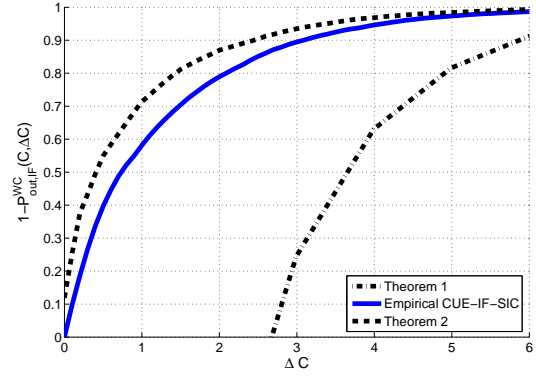}
\end{center}
\caption{Theorem 1 (upper bound on outage probability/lower bound on achievable rate) and Theorem 2 (lower bound on WC outage probability/upper bound on achievable rate) for $N_r\times 2$ MIMO channels ($N_r\geq2$) with mutual information $C_{}=14$.}
\label{fig:xxx}
\end{figure}

\subsection{Space-Time Precoding}

Hitherto the role of random precoding
was limited to facilitating performance evaluation.
Namely, applying CUE precoding results
in performance being dictated solely by the singular values of the channel, so that one can then consider the worst case performance only with respect to the latter.

In contrast, applying random precoding over time as well as space has operational significance, allowing to improve the guaranteed performance as we quantify next.

\subsubsection{Background}


A block of $T$ channel uses is processed jointly so that the $N_r \times N_t$ physical MIMO channel (\ref{eq:channel_model}) is transformed into a ``time-extended" $N_r T \times  N_t T$ MIMO channel.
A unitary precoding matrix $\svv{P}_{st,c} \in \mathbb{C}^{\Nt T \times \Nt T}$ that can be either deterministic or random is then
applied to the time-extended channel.
At the receiver, IF equalization is
employed.

 Hence, the equivalent channel takes the form
\begin{align}
 \mathcal{H}_c&=\svv{I}_{T\times T} \otimes \svv{H}_c,
 \end{align}
 where $\otimes$ is the Kronecker product. Let $\bar{\boldsymbol{x}}_c\in \mathbb{C}^{\Nt T \times 1}$ be the input vector to the time-extended channel. It follows that the output of the time-extended channel is given by
 \begin{align}
\bar{\boldsymbol{y}}_c^P=\mathcal{H}_c\svv{P}_{st,c}\bar{\boldsymbol{x}}_c+\bar{\boldsymbol{z}}_c,
\end{align}
where $\bar{\boldsymbol{z}}_c$ is i.i.d.  unit-variance circularly symmetric complex Gaussian noise.
As we assume that the precoding matrix is unitary (for both deterministic or random cases), the WI mutual information of this channel (normalized per channel use) remains unchanged, i.e.,
\begin{align}
    \frac{1}{T}\log\det\left(\svv{I}+(\mathcal{H}_c\svv{P}_{st,c})(\mathcal{H}_c\svv{P}_{st,c})^H\right)= C_{}.
\end{align}
When using a given space-time precoding ensemble, the WC scheme outage is defined as
\begin{align}
     & P^{\rm WC,scheme}_{\rm out}(C_{},\Delta C)=  \nonumber \\
    & \sup_{\svv{H}_c\in\mathbb{H}(C_{})}\prob\left(\frac{1}{T}R_{\rm scheme}(\mathcal{H}_c\cdot\svv{P}_{st,c})<C-\Delta C\right).
     \label{eq:P_ST_WC_OUT}
\end{align}
\subsubsection{Upper Bound}
\label{sec:st_upper}
For $N_t=2$, space-time CUE precoding results in a $N_r T\times 2T$ MIMO channel. An upper bound on the WC outage probability can be obtained from Theorem 1 in \cite{DomanovitzE16}, by substituting $\Nt=2 T$.
\subsubsection{Lower Bound}
\label{sec:st_lower}
Define

\begin{align}
\mathcal{B}_1(T,k,R,\rho_1,\rho_2)&=\left\{\underline{\lambda}:\prod_{i=1}^{k}\left(1+\rho_1\lambda_i+\rho_2(1-\lambda_i)\right)\leq2^{R\frac{k}{2}}\right\} \nonumber \\
\mathcal{B}_2(T,k,\tilde{R},\rho_1,\rho_2)&=\left\{\underline{\lambda}:\prod_{i=1}^{2T-k}\left(1+\rho_1\lambda_i+\rho_2(1-\lambda_i)\right)\leq2^{\tilde{R}}\right\} \nonumber \\
\kappa_{m_1,m_2,n}&=\prod_{j=1}^{n}\frac{\Gamma(m_1-n+j)\Gamma(m_2-n+j)\Gamma(1+j)}{\Gamma(2)\Gamma(m_1+m_2-n+j)}.
\end{align}

where $\tilde{R}=\frac{k}{2}\max(R-(k-T),0)$.

\begin{theorem}
For an $N_r\times 2$ compound channel with WI-MI equal $C_{}$, and CUE precoding over $T$ time extensions, we have
\begin{align}
& P^{\rm WC}_{\rm out}(C_{},R)\geq \max_{0\leq \rho_2\leq 2^{C_{}}/2}\max_k P_{\rm out}
\end{align}
where $P_{\rm out}=P_{\rm out}(k,T,R,\rho_1,\rho_2)$ and
\begin{itemize}\addtolength{\itemsep}{0.5\baselineskip}
\item
For $1\leq k \leq T$:
\begin{align}
P_{\rm out}=\kappa_1\int\displaylimits_{\mathcal{B}_1}\prod_{i=1}^{k}\lambda_i^{T-k}(1-\lambda_i)^{T-k}\prod_{i<j}(\lambda_j-\lambda_i)^2d\underline{\lambda}
\end{align}
\normalsize
\item
For $T+1\leq k \leq 2T$:
\begin{align}
 P_{\rm out}=\kappa_2\int\displaylimits_{\mathcal{B}_2}\prod_{i=1}^{2T-k}\lambda_i^{k-T}(1-\lambda_i)^{k-T}\prod_{i<j}(\lambda_j-\lambda_i)^2d\underline{\lambda}.
\end{align}
\normalsize
\end{itemize}
and where \mbox{$\kappa_1=\kappa_{T,T,k}^{-1}$} and \mbox{$\kappa_2=\kappa_{T,T,2T-k}^{-1}$}.
\label{thm:thm3}
\end{theorem}
\begin{proof}
The proof depends on the eigenvalue distribution of submatrices of $\svv{P}_c$. As mentioned above, these eigenvalues follow the Jacobi distribution. The full description of the distribution and proof appears in Appendix~\ref{sec:appa}
\end{proof}

\subsubsection{Comparison of Bounds and Empirical Results}
We compare the obtained upper and lower bounds with the empirical performance results of CUE-precoded IF-SIC. In addition, for a  $T=2$
time-extended $N_r\times 2$ channel, it is natural to also compare  performance with that obtained by replacing CUE precoding with algebraic precoding. Specifically, we consider  Alamouti  and golden code
precoding.\footnote{When using  a fixed space-time precoding matrix,  we apply in addition CUE precoding to the physical channel.}

To that end, let us define the $\varepsilon$-outage capacity of a scheme $R_{\rm scheme}(\svv{P}_{st,c};\varepsilon)$ as the rate  for which
\begin{align}
     P^{\rm WC,scheme}_{\rm out}\left(C_{},R_{\rm scheme}(\svv{P}_{st,c};\varepsilon)\right)=\varepsilon.
\end{align}
Further, the guaranteed transmission efficiency of a scheme, at a given outage probability $\varepsilon$ and WI mutual information $C_{}$, is defined as
\begin{align}
    \eta_{\varepsilon}(C_{},\svv{P}_{st,c})=
     R_{\rm scheme}(\svv{P}_{st,c};\varepsilon) / C_{}.
\end{align}
Figure~\ref{fig:eff_1per_2x2} depicts the guaranteed efficiency at $1\%$ outage for several precoding options for an $N_r\times 2$ channel and $T=1,2$. We plot the empirical efficiency for both IF-SIC and ML receivers. It can be seen that for CUE precoding,
the performance of IF-SIC is very close to that of ML.

\begin{figure}
\begin{center}
\includegraphics[width=0.96\columnwidth]{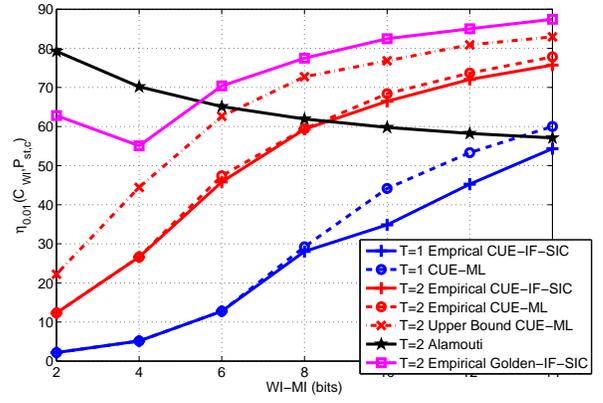}
\end{center}
\caption{Guaranteed efficiency at $1\%$ outage probability for the $N_r \times 2$ MIMO channel with various precoding and decoding options, for $T=1,2$.}
\label{fig:eff_1per_2x2}
\end{figure}

\begin{figure}
\begin{center}
\includegraphics[width=0.96\columnwidth]{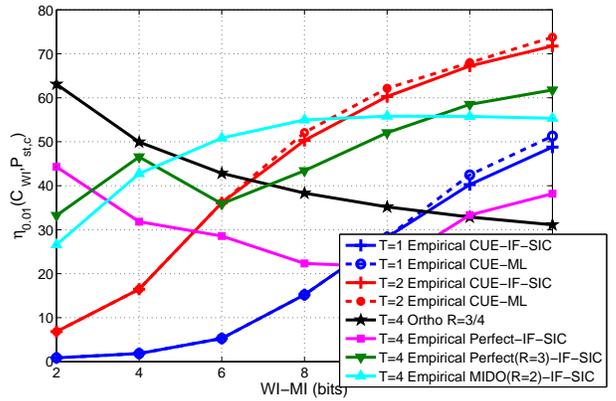}
\end{center}
\caption{Guaranteed efficiency at $1\%$ outage probability for the $N_r\times 4$ MIMO channel with various precoding and decoding options, for $T=1,2,4$.}
\label{fig:eff_1per_4x4}
\end{figure}

We also present empirical results for an $N_r\times 4$ channel.
Figure~\ref{fig:eff_1per_4x4} depicts the guaranteed efficiency at $1\%$ for several precoding and receiver topologies, where the algebraic codes considered
are orthogonal space-time block precoding (rate $3/4$),  the perfect code \cite{oggier2006perfect}, the latter punctured to rate $3$, and also the MIDO  (rate $2$) code \cite{OggierMIDO:2010}.

\section{Application: Closed-Form Bounds for the Symmetric-Rate Capacity of the Rayleigh Fading Multiple-Access Channel}
The lower bound derived in the previous section can be easily shown to cover also the case of a $2\times \Nt$ channel where $\Nt\geq2$. In this section we  adapt the bound for the case of a $1\times \Nt$ system where in this case we are interested in a MAC scenario, that is the encoders correspond to different (and distributed) users. More specifically, we analyze the ML performance of a Rayleigh fading MAC where all terminals are equipped with a single antenna and where we consider a simple transmission protocol
as described below.


The channel is described by
\begin{align}
    y=\sum_{i=1}^{\Nt}h_i x_i+n
\end{align}
where ${h}_{i}\sim \sqrt{\SNR} \cdot  \mathcal{CN}(0,1)$  and $n\sim\mathcal{CN}(0,1)$, and where there is no statistical dependence between any of these random variables. Without loss of generality we assume throughout the analysis to follow that $\SNR=1$, i.e., we absorb the $\SNR$ into the channel.

The capacity region of the channel is given by \cite{CoverBook} all rate vectors
$(R_1, \ldots, R_{N_t})$ satisfying \begin{align}
\sum_{i\in S} R_i < \log\left(1+\sum_{i\in S}|{ h}_i|^2\right),
\label{eq:capacity_region_MAC}
\end{align}
for all ${S\subseteq\{1, 2,\ldots,\Nt\}}$.
We denote the sum capacity by
\begin{align}
C=\log\left(1+\sum_{i=1}^{N_t}|{ h}_i|^2\right).
\label{eq:sum_capacity_MAC}
\end{align}

If we impose the constraint that all users transmit at the same rate, then the maximal achievable symmetric-rate is given by substituting $R_i=C_{\rm sym}/N_t$ in \eqref{eq:capacity_region_MAC}, from which it follows that the symmetric-rate capacity is dictated by the bottleneck:
\begin{align}
C_{\rm sym}=\min_{S\subseteq\{1, 2,\ldots,\Nt\}}\frac{\Nt}{|S|}\log\left(1+\sum_{i\in S}|{ h}_i|^2\right).
\label{eq:sym_capacity_MAC}
\end{align}
Note that \eqref{eq:sym_capacity_MAC} is a special case of (\ref{eq:R_ML_spaceOnly}).

While in general applying a CUE precoding transformation ${\svv{P}}_c$ implies joint processing at the encoders, which is precluded in a MAC setting, in an i.i.d. Rayleigh fading environment, this random transformation is actually performed by nature.\footnote{This follows since the left and right singular vector matrices of the an i.i.d. Gaussian matrix $\svv{H}_c$ are equal to the eigenvector matrices of the  Wishart ensembles $\svv{H}_c\svv{H}_c^{H}$ and $\svv{H}^H_c\svv{H}_c$, respectively. The latter are known to be CUE (Haar) distributed. See, e.g., Chapter~4.6 in \cite{edelman2005random}.} Hence
the results developed in the previous sections readily apply to this scenario.

We next analyze the conditional ``cumulative distribution function" \begin{align}
    \prob(C_{\rm sym}<R| C)
    \label{eq:cum}
\end{align}
for i.i.d. Rayleigh fading.\footnote{We use quotation marks since we impose strict inequality in $C_{\rm sym}<R$.} The latter quantity gives a full statistical
characterization for the performance of a transmission protocol
where all users transmit at a rate just below the equal-rate capacity (per user) of the channel, where the underlying assumption is that this rate is dictated to the users by the base station.

Another interpretation for \eqref{eq:cum} is as a conditional outage probability in an open-loop scenario. That is, in a scenario where
all users (when they are active) transmit at a common target rate $R_{\rm tar}$, for a given  number of active users is $N_t$, then the outage probability is given by $E[ \prob(C_{\rm sym}<N_t \cdot R_{\rm tar}| C)]$ where the expectation is over $C$ and is computed w.r.t. an i.i.d. Rayleigh distribution.

We will obtain tight bounds on the distribution of the rate attained by such a transmission scheme. In particular, these bounds characterize the probability that the dominant face of the MAC capacity
region contains an equal-rate point, i.e., that the scheme strictly attains the sum-capacity of the
channel. The analysis provides a non-asymptotic counterpart to the diversity-multiplexing tradeoff of the MAC channel and can
also serve to obtain bounds on the ergodic capacity of the described protocol.
We start by analyzing the simplest case of a two-user MAC.
\begin{theorem}
For a two-user i.i.d. Rayleigh fading MAC, we have
    \begin{align}
        \prob(C_{\rm sym}<R| C) = 2\cdot\frac{2^{R/2}-1}{2^{C}-1}
    \end{align}
    \label{thm:thm4}
\end{theorem}
\begin{proof}
Given $C$, $\svv{h} \triangleq (h_1,h_2)$ is uniformly distributed over a two-dimensional  complex sphere of radius
$\sqrt{2^ {C}-1}$. Hence, $\svv{h}/ \|\svv{h} \|$ can be viewed as the first row of a unitary matrix $\svv{U}$ drawn from the CUE.


By \eqref{eq:sym_capacity_MAC} and using the notation of
\eqref{eq:joint,st}, we obtain (cf. \eqref{eq:Rj_k123})
\begin{align}
    C_{\rm sym}=\min\left\{C(\{1\}),C(\{2\}),C_{}\right\}.
\end{align}
We start by analyzing $C(\{1\})$ which is given by
\begin{align}
C(\{1\})&=2\log\left(1+|h_1|^2\right) \nonumber \\
& = 2\log\left(1+|\svv{U}_{1,1}|^2(2^{C}-1)\right).
\end{align}
It follows that
\begin{align}
    \prob(C(\{1\})<R|C)=\prob\left(|\svv{U}_{1,1}|^2<\frac{2^{R/2}-1}{2^C-1}\right).
\end{align}
Since (see, e.g., \cite{NarulaTrottWornell:1999}) for a  $2\times2$ CUE matrix, we have $|\svv{U}_{1,1}|^2\sim \rm{Unif}([0,1])$, we obtain
\begin{align}
    \prob(C(\{1\})<R|C)=\frac{2^{R/2}-1}{2^C-1}.
\end{align}
We note that, similarly to Theorem~\ref{thm:thm1}, it can be shown that the events $\{C(\{1\})<R\}$
and $\{C(\{2\})<R\}$ are disjoint. Since by symmetry we
further have that
\begin{align}
    \prob(C(\{1\})<R|C)=\prob(C(\{2\})<R|C),
\end{align}
it follows that
\begin{align}
    \prob(C_{\rm sym}<R| C) = 2\cdot\frac{2^{R/2}-1}{2^{C}-1}.
    \label{eq:nonzeroprob}
\end{align}
\end{proof}

We note that the probability in \eqref{eq:nonzeroprob}
is strictly smaller than one at $R=C$.
Furthermore,  the probability that the symmetric-rate capacity
coincides with the sum capacity (i.e., that the equal-rate line
passes through the dominant face of the capacity region) is given
by
\begin{align}
    \prob\left({C_{\rm sym}=C|C}\right)&=1-\prob\left({C_{\rm sym}<C|C}\right) \nonumber \\
    & = 1 - 2\cdot\frac{2^{C/2}-1}{2^{C}-1}.
    \label{eq:green}
\end{align}
Note that the latter probability tends to one exponentially
fast in $C$.


Figure~\ref{fig:capRegion2users} depicts the capacity region for several cases where the sum-capacity equal $2$ bits/channel use.

The probability that the capacity region is of the type of the dashed
line (equal-rate line passing through the dominant face of the region)
is given by \eqref{eq:green}
which for $C=2$ yields
\begin{align}
    \prob\left({C_{\rm sym}=C|C=2}\right)
    & = 1 - 2\cdot\frac{2^1-1}{2^2-1} \nonumber \\
    & = 1 - 2\cdot\frac{1}{3} \nonumber \\
    &=\frac{1}{3}.
    \label{eq:eq}
\end{align}
Figure~\ref{fig:pdf_c2} depicts
the probability density function of the symmetric-rate capacity
of a two-user i.i.d. Rayleigh fading MAC given that the sum-capacity is $C=2$. The probability in \eqref{eq:eq} manifests itself as a delta function at the sum-capacity.

\begin{figure}
\begin{center}
\includegraphics[width=1\columnwidth]{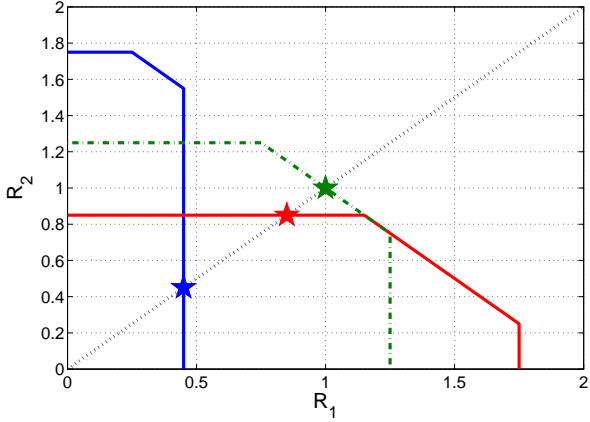}
\end{center}
\caption{Different capacity regions corresponding to a two-user  MAC with sum-capacity $C=2$. For the channel depicted with a dashed-dotted line, the dominant face constitutes the bottleneck and $C_{\rm sym}=C$. }
\label{fig:capRegion2users}
\end{figure}

\begin{figure}
\begin{center}
\includegraphics[width=1\columnwidth]{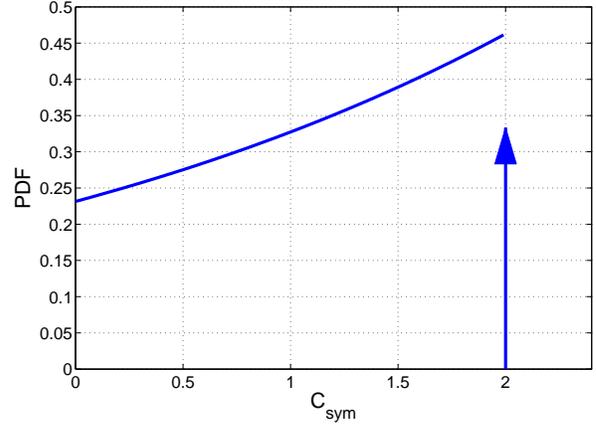}
\end{center}
\caption{Probability density function of the symmetric-rate capacity
of a two-user i.i.d. Rayleigh fading MAC given that the sum-capacity is $C=2$.}
\label{fig:pdf_c2}
\end{figure}

Theorem~\ref{thm:thm4} may be extended to the case of $\Nt>2$ but
rather than obtaining an exact characterization of the distribution of the symmetric-rate capacity, we will derive lower and upper bounds for it.
We begin with the following lemma from which Theorem~\ref{thm:thm5} follows.
\begin{lemma}
For an $N_t$-user i.i.d. Rayleigh fading  MAC with sum capacity $C$,
for any  subset of users $S\subseteq\{1, 2,\ldots,\Nt\}$ with cardinality $k$, we have
 \begin{align}
     \prob\left(C(\{S\})<R_{}|C_{}\right)
     &= \frac{\mathcal{B}(\frac{2^{R_{}|S|/\Nt}-1}{2^{C_{}}-1};|S|,\Nt-|S|)}{\mathcal{B}(1;|S|,\Nt-|S|)}\nonumber \\
     &\triangleq P_{\rm out}(k,R_{}|C)
 \end{align}
 where $0\leq R_{} \leq C_{}$ and
 \begin{align}
     \mathcal{B}(x;a,b)=\int_0^{x}u^{a-1}(1-u)^{b-1}du
 \end{align}is the incomplete beta function.
 \label{lem:lem1}
\end{lemma}
\begin{proof}
Similar to the case of two users,  $\svv{h}\triangleq(h_1,\ldots,h_{\Nt})$ is uniformly distributed over
an $N_t$-dimensional complex sphere of radius $\sqrt{2^ {C}-1}$ and hence $\svv{h}/\| \svv{h} \|$ may be viewed as the first row of a unitary matrix $\svv{U}$ taken from the CUE.

By symmetry, for any set $S$ with cardinally $k$, the distribution of
$ C(\{S\})$ is equal to that of
\begin{align}
    C(\{1,2,\ldots,k\})
    &=\frac{\Nt}{k}\log\left(1+\sum_{i=1}^k |h_i|^2\right) \nonumber \\
                  &=\frac{\Nt}{k}\log\left(1+(2^C-1)\sum_{i=1}^k |U_{1,i}|^2\right)
  \label{eq:58}
\end{align}

Denoting the partial sum of $k$ entries
as $X=\sum_{i=1}^k |U_{1,i}|^2$, we therefore have
\begin{align}
   \prob(C(\{S\})<R|C)&=\prob\left( 1+ \left(2^C-1\right)X <2^{R\frac{\Nt}{k}}\right)
   \nonumber \\
   & = \prob\left(X<\frac{2^{R\frac{\Nt}{k}}-1}{2^C-1}\right).
\end{align}
We note that the vector
$( |\svv{U}_{1,1}|^2,\ldots,|\svv{U}_{1,N_t}|^2 )$ has the Dirichlet distribution and a partial sum of its entries has a Jacobi distribution.
To see this, we note that \eqref{eq:58} can be written as
\begin{align}
    C(\{1,2,\ldots,k\})=\frac{\Nt}{k}\log\left(1+(2^C-1)\svv{U}(k)_{1}\svv{U}(k)^H_{1}\right)
 \end{align}
where $\svv{U}(k)_{1}$ is a vector which holds the first $k$ elements of the first row of $\svv{U}$. Noting that since $\svv{U}(k)_{1}$ is a submatrix of a unitary matrix, its singular values follow the Jacobi distribution. It follows that $X$ has Jacobi distribution with rank 1.


We thus obtain
(using e.g., \cite{dar2013jacobi}),
\begin{align}
     \prob\left(C(\{S\})<R_{}|C_{}\right) &= \int_{0}^{\frac{2^{R_{}k/\Nt}-1}{2^{C_{}}-1}}x^{k-1}x^{\Nt-k-1}d\lambda \nonumber \\ &=\frac{\mathcal{B}(\frac{2^{R_{}k/\Nt}-1}{2^{C_{}}-1};k,\Nt-k)}{\mathcal{B}(1;k,\Nt-k)},\nonumber
 \end{align}
where
\begin{align}
     \mathcal{B}(x;a,b)=\int_0^{x}u^{a-1}(1-u)^{b-1}du \nonumber
 \end{align}
 is the incomplete beta function.

\end{proof}


\begin{theorem}
For an $N_t$-user Rayleigh MAC, we have
{
\begin{align}
 \max_k P_{\rm out}(k,R_{}|C) & \leq  \prob\left(C_{\rm sym} <R_{}|C\right) \\
 & \leq  \sum_{k=1}^{\Nt} {{\Nt}\choose{k}} P_{\rm out}(k,R_{}|C). \nonumber
 \end{align}
 }
where $P_{\rm out}(k,R_{}|C)$ is defined Lemma~\ref{lem:lem1}.
\label{thm:thm5}
\end{theorem}
\begin{proof}
To establish the left hand side of the theorem, first note
that $C_{\rm sym}\leq C(\{S\})$ for any subset $S$ and hence
%
\begin{align}
C_{\rm sym} \leq \min_k C(\{1,2,\ldots,k\}).
\end{align}
It follows that
\begin{align}
 \prob\left(C_{\rm sym} <R_{}|C\right) &\geq \prob\left( \min_k C(\{1,2,\ldots,k\})<R_{} \Big|  C\right)   \nonumber \\
 &= \prob\left( \bigcup_k \left\{ C(\{1,2,\ldots,k\})<R_{} \right\}\Big|  C\right) \nonumber \\
 &\geq \max_k \prob\left(  C(\{1,2,\ldots,k\})<R_{}  \Big|  C\right) \nonumber \\
 & = \max_k P_{\rm out}(k,R_{}|C).
\end{align}
The right hand side is proved by applying the union bound.
\end{proof}

Figure~\ref{fig:lowUpBoundOutageMAC} depicts these bounds for a $4$-user i.i.d. Rayleigh fading  MAC with sum capacity $C_{}=8$.

\begin{figure}
\begin{center}
\includegraphics[width=1\columnwidth]{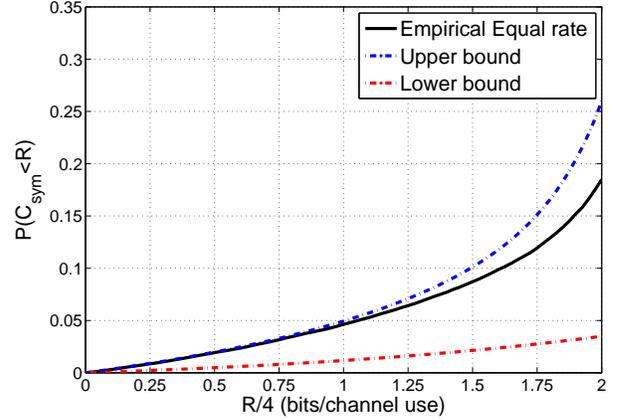}
\end{center}
\caption{Comparison of empirical evaluation of \eqref{eq:sym_capacity_MAC} and Theorem 5 (upper  and lower bounds for the outage probability) for a $1 \times 4$ i.i.d. Rayleigh fading MAC with sum capacity $C_{}=8$.}
\label{fig:lowUpBoundOutageMAC}
\end{figure}

\subsection{Performance of Integer-Forcing over the Multiple-Access Channel}
It has been further shown in \cite{IntegerForcing} that the IF receiver achieves the diversity-multiplexing tradeoff (DMT) over i.i.d. Rayleigh fading channels where the number of receive antennas is greater or equal to the number of transmit antennas.

We observe that this does not hold in the general case; in particular, for the case of a single receive antenna.
Specifically, Figure~\ref{fig:logErr_ML_IF_T2} depicts (in logarithmic scale) the empirical outage probability of the IF  receiver and the exact outage probability of Gaussian codebooks and an ML receiver (as given by Theorem~\ref{thm:thm4}) for a two-user i.i.d. Rayleigh fading MAC. It is evident that the slopes are different. This raises the question of whether IF is inherently ill-suited for the MAC channel. A negative answer to this question may be inferred by recalling some lessons from the MAC DMT.

While the optimal DMT for the i.i.d. Rayleigh fading MAC was
derived in \cite{2004:DMT_MAC} by considering Gaussian codebooks of
sufficient length, it was subsequently shown that the MAC DMT can be achieved by structured codebooks by combining uncoded QAM constellations with
space-time unitary precoding (and ML decoding). Specifically, such a  MAC-DMT achieving  construction is given in  \cite{2011:Hollanti_DMT_optimal_codes}.
This suggests that the sub-optimality of the IF receiver observed in
 Figure~\ref{fig:logErr_ML_IF_T2} may at least in part be remedied by applying unitary space-time precoding at each of the transmitters.
 We note  that each transmitter applies precoding only to its own data streams so the distributed nature of the problem is not violated.

Following this approach, we have implemented the IF receiver with  unitary space-time precoding applied at each transmitter. We have employed random CUE precoding over two ($T=2$) time instances as well as deterministic precoding using the
matrices proposed in
\cite{badr2008distributed}.\footnote{When using an ML receiver, this space-time code is known to achieve the DMT for multiplexing rates $r\leq\frac{1}{5}$. As detailed in \cite{2011:Hollanti_DMT_optimal_codes}, whether this code achieves the optimal MAC-DMT
also when $r> \frac{1}{5}$ remains an open question.}
These matrices can be expressed as
\begin{align}
    \svv{P}_{st,c}^1&=\frac{1}{\sqrt{5}}
    \begin{bmatrix}
    \alpha & \alpha\phi \\
    \bar{\alpha} & \bar{\alpha}\bar{\phi}
    \end{bmatrix} \nonumber \\
     \svv{P}_{st,c}^2&=\frac{1}{\sqrt{5}}
    \begin{bmatrix}
    j\alpha & j\alpha\phi \\
    \bar{\alpha} & \bar{\alpha}\bar{\phi}
    \end{bmatrix}
    \label{eq:Badr_Bel}
\end{align}
where
\begin{align}
    \phi&=\frac{1+\sqrt{5}}{2} \nonumber \\
    \bar{\phi}&=\frac{1-\sqrt{5}}{2} \nonumber \\
    \alpha&=1+j-j\phi \nonumber \\
    \bar{\alpha}&=1+j-j\bar{\phi}.
\end{align}
As can be seen, both random space-time precoding and the precoding matrices in \eqref{eq:Badr_Bel} improve the outage probability for most target rates.

\begin{figure}
\begin{center}
\includegraphics[width=1\columnwidth]{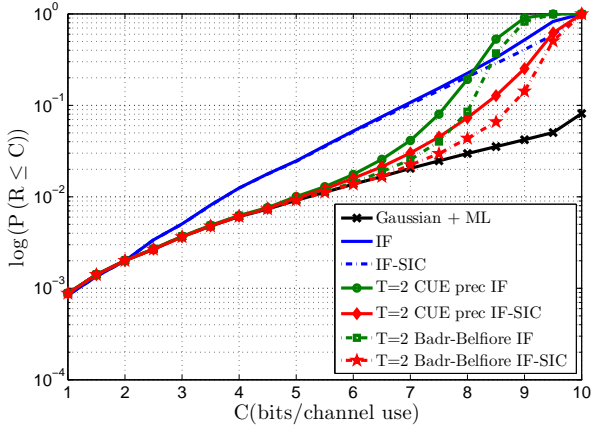}
\end{center}
\caption{Outage probability of  linear codes (with and without space-time precoding) with IF equalization versus  Gaussian codebooks with ML decoding for a $1 \times 2$ i.i.d. Rayleigh fading  MAC with sum capacity $C_{}=10$.}
\label{fig:logErr_ML_IF_T2}
\end{figure}

Figure~\ref{fig:Rout_Csum_ML_IF_T2} depicts the fraction of the ergodic capacity achieved when all users transmit at the symmetric-rate capacity (per user) for the two-user i.i.d. Rayleigh fading MAC versus the fraction achieved when using linear codes (at the maximal achievable rate) in conjunction with IF-SIC  with different precoding methods as described above.
We observe that IF-SIC combined with space-time precoded linear codes
achieves a large fraction of the ergodic symmetric-rate capacity. Furthermore, the fraction of ergodic capacity achieved approaches one as the sum-capacity grows.

\begin{figure}
\begin{center}
\includegraphics[width=1\columnwidth]{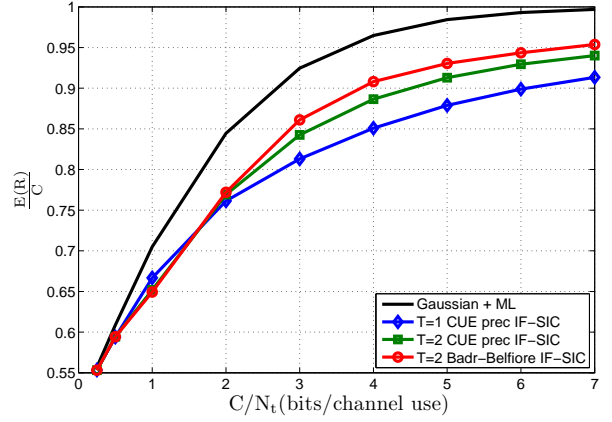}
\end{center}
\caption{Ergodic capacity conditioned on sum-capacity of linear codes (with and without space-time precoding) with IF equalization versus Gaussian codebooks with ML decoding over a  $1 \times 2$  i.i.d. Rayleigh fading MAC.}
\label{fig:Rout_Csum_ML_IF_T2}
\end{figure}

\section{Discussion and Outlook}
For the $N_r\times 2$ compound MIMO channel, using CUE precoding over a time-extend channel offers significant benefit over space-only precoding. However, space-time CUE precoding falls short when compared to algebraic space-time precoding.
Specifically, the combination of Alamouti precoding at low rates and golden code precoding (with IF-SIC) at high rates is superior to CUE precoding.

Nonetheless, for the $N_r\times 4$ compound MIMO channel, we observe from the empirical results (Figure~\ref{fig:eff_1per_4x4}) that there is a region where using random space-time CUE precoding results in the highest guaranteed efficiency.
This provides motivation for searching for fixed precoding matrices that
yield better results than perfect codes at the price of a small outage probability.

As a concluding remark, we note that the derived lower bound  holds only for the case of a maximum of two distinct singular values in the SVD decomposition of $\svv{H}_c$.
Nevertheless, the treatment holds also for the important case of a open-loop MAC channel with a single
receive antenna, where the transmitters (also equipped with a single antenna)
know only the sum capacity of the channel, where the results are modified as described in footnote~\ref{footnote5}.


\begin{appendices}
\section{Proof of Theorem~\ref{thm:thm3}}
\label{sec:appa}
\begin{proof}
Using $T$ time extensions we arrive at a $N_r T\times 2T$ equivalent channel. The ML lower bound (\ref{eq:R_ML_spaceOnly}) now takes the form
\begin{align}
R_{\rm ML,ST} &=\frac{1}{T}\min_k\frac{2T}{k}\min_{S\in S_k}\log\det\left(\svv{I}_k+\svv{H}_{s}\svv{H}_{s}^H\right) \nonumber \\
&=\min_k \min_{S\in S_k}\frac{2}{k} C(S\in S_k),
\end{align}
where $S_k$ is the set which has all the subsets of cardinality $k$ contained in $\{1,2,\cdots,2 T\}$. Hence $\svv{H}_{s}$ is a submatrix of $\mathcal{H}_c\svv{P}_{st,c}$ formed by taking $k$ columns.

Using (\ref{eq:SVDofHc}), and after possibly applying column permutations, the effective channel takes the form
\begin{align}
\svv{E}\triangleq
\widetilde{\svv{U}_c}
\begin{bmatrix}
\sqrt{\rho_1}\svv{I}_{T\times T}  & \svv{0} \\ \svv{0} & \sqrt{ \rho_2} \svv{I}_{T\times T}
\end{bmatrix}
\svv{P}_{st,c},
\label{eq:PstcSVD}
\end{align}
where $\widetilde{\svv{A}}=\svv{I}_{T\times T} \otimes \svv{A}$. Since $\svv{P}_{st,c}$ is drawn from the CUE, it follows that $\svv{E}$ is equal to $\mathcal{H}_c\svv{P}_{st,c}$ in distribution and thus we assume that the
channel is the former for sake of analysis.
In particular, we have
\begin{align}
C(S\in S_k)&=\log\det\left(\svv{I}_{N_r T}+\svv{E}_{s}\svv{E}_{s}^H\right) \nonumber \\
&=\log\det\left(\svv{I}_k+\svv{E}_{s}^H\svv{E}_{s}\right).
\end{align}
Let us use the notation $\left[~\right]_{s}$ to denote the matrix
resulting from a specific selection of $k$ columns from a matrix, corresponding to the chosen set $s$. Denoting
\begin{align}
\svv{E}_s =   \left[\svv{E}\right]_{s}  \nonumber
=\widetilde{\svv{U}_c}
\begin{bmatrix}
\sqrt{\rho_1}\svv{I}_{T\times T}  & \svv{0} \\ \svv{0} & \sqrt{ \rho_2} \svv{I}_{T\times T}
\end{bmatrix}
\left[\svv{P}_{st,c}\right]_s,
\end{align}
we have
\begin{align}
&C(S\in S_k)=\log\det\left(\svv{I}_{k}+\svv{E}_{s}\svv{E}_{s}^H\right)\nonumber \\
&=\log\det\left(\svv{I}_k+\left[\svv{P}_{st,c}\right]_s^H\begin{bmatrix}
{\rho_1}\svv{I}_{T\times T}  & \svv{0} \\ \svv{0} & {\rho_2} \svv{I}_{T\times T}
\end{bmatrix}\left[\svv{P}_{st,c}\right]_s\right) \nonumber \\ &=\log\det\left(\svv{I}_k+\rho_1\svv{P}_1^H\svv{P}_1+\rho_2\svv{P}_2^H\svv{P}_2\right),
\end{align}
where $\left[\svv{P}_{st,c}\right] = \begin{bmatrix} \svv{P}_1 \\ \hline \svv{P}_2
\end{bmatrix}$.

As described in \cite{dar2013jacobi}, we note that the singular values of  $\svv{P}_1$ (which is a rectangular submatrix of dimensions $\frac{T}{2}\times k$ of the $2 T\times 2 T$ unitary matrix $\svv{P}_{ST,c}$) has the following Jacobi distribution
\begin{itemize}
\item
    When $1\leq k\leq T$, the singular values of $\svv{P}_1$ have the same distribution as the eigenvalues of
        the Jacobi ensemble $\mathcal{J}(T,T,k)$.
\item
    When $T < k \leq 2T$, using Lemma~1 in \cite{dar2013jacobi}, we have that the singular values of $\svv{P}_1$ have the same distribution as the eigenvalues of the Jacobi ensemble $\mathcal{J}(T,T,2T-k)$.
\end{itemize}

Further, we recall a derivation appearing in Lemma 1 of \cite{dar2013jacobi} (which is a corollary of \cite{paige1981towards}) and note that since $\svv{P}_{st,c}$ is unitary, we have
\begin{align}
\svv{P}_1^H\svv{P}_1+\svv{P}_2^H\svv{P}_2=\svv{I}_k.
\end{align}
Therefore
\begin{align}
\svv{P}_1^H\svv{P}_1&=\svv{I}_k-\svv{P}_2^H\svv{P}_2 \nonumber \\
\svv{U}\svv{D}_1\svv{V}^H&=\svv{I}_k-\svv{P}_2^H\svv{P}_2 \nonumber \\
\svv{D}_1&=\svv{U}^H\left(\svv{I}_k-\svv{P}_2^H\svv{P}_2\right)\svv{V} \nonumber \\
\svv{D}_1&=\svv{D}_2.
\end{align}

Let $\{\lambda_i^{(1)}\}_{i=1}^{k}$ and $\{\lambda_i^{(2)}\}_{i=1}^{k}$ be the eigenvalues of $\svv{P}_1^H\svv{P}_1$ and $\svv{P}_2^H\svv{P}_2$, respectively. It follows that
\begin{align}
    \lambda_i^{(2)}=1-\lambda_i^{(1)},
\end{align}
and hence
\begin{align}
&C(S\in S_k)=\det\left(\svv{I}_k+\rho_1\svv{P}_1^H\svv{P}_1+\rho_2\svv{P}_2^H\svv{P}_2\right) \nonumber \\
&=\prod_{i=1}^{m}\left(1+\rho_1\lambda_i^{(1)}+\rho_2(1-\lambda_i^{(1)})\right).
\label{eq:valOfDet}
\end{align}

Therefore, for a specific choice of columns $s \in S_k$,  the outage probability may be written as
\begin{align}
& P_{\rm out}\left(\frac{2}{k}C(S\in S_k)\right) \nonumber \\
& = \prob\left(\frac{2}{k}C(S\in S_k)<R\right)\nonumber \\
& =\prob\left(\log\det\left(\svv{I}_k+\rho_1\svv{P}_1^H\svv{P}_1+\rho_2\svv{P}_2^H\svv{P}_2\right)<R\frac{k}{2}\right) \nonumber \\
& = \prob\left(\prod_{i=1}^{k}\left(1+\rho_1\lambda_i^{(1)}+\rho_2(1-\lambda_i^{(1)})\right)<2^{R\frac{k}{2}}\right) \nonumber \\
& = \prob\left(\prod_{i=1}^{k}\left(1+\rho_2+\lambda_i^{(1)}(\rho_1-\rho_2)\right)<2^{R\frac{k}{2}}\right).
\end{align}

Without loss of generality, we assume that $\rho_2\leq\rho_1$ and hence $0\leq \rho_2\leq 2^{C_{}}/2$. Using the explicit expression for the Jacobi distribution (\ref{eq:jacobi}) of these singular values, we have
\begin{itemize}
\item
For $1 \leq k\leq T$:
\begin{align}
    &P_{\rm out}(S\in S_k)=P_{\rm out}(k,T,R,\rho_1,\rho_2) \nonumber \\
    &=\svv{K}_1\int\displaylimits_{\mathcal{B}_1}\prod_{i=1}^{k}\lambda_i^{T-k}(1-\lambda_i)^{T-k}\prod_{i<j}(\lambda_j-\lambda_i)^2d\underline{\lambda}.
\end{align}
\item
For $T \leq k\leq 2T$, by Theorem~3 of \cite{dar2013jacobi}, we have
\begin{align}
P_{\rm out}(k,T,R,\rho_1,\rho_2)=P_{\rm out}(2T-k,2T,\tilde{R},\rho_1,\rho_2),
\end{align}
and thus
\begin{align}
    &P_{\rm out}(S\in S_k)=P_{\rm out}(2T-k,2T,\tilde{R},\rho_1,\rho_2) \nonumber \\
    &=\svv{K}_2\int\displaylimits_{\mathcal{B}_2}\prod_{i=1}^{2T-k}\lambda_i^{k-T}(1-\lambda_i)^{k-T}\prod_{i<j}(\lambda_j-\lambda_i)^2d\underline{\lambda}.
    \end{align}
\end{itemize}

Defining $S_{1,k}$ to be the first element in $S_k$, we have
\small
\begin{align}
&\prob\left(\min_k \min_{S\in S_k} \frac{2}{k} C(S\in S_k)<R\right) \nonumber \\
&\geq\prob\left(\frac{2}{k} \min_k C(S_{1,k})<R\right) \nonumber \\
&\geq \max_k \prob\left( \frac{2}{k} C(S_{1,k})<R\right) \nonumber \\
& = \max_k P_{\rm out}(S_{1,k}).
\end{align}
\normalsize
To conclude, we have established that
\begin{align}
P^{\rm WC}_{\rm out}(C_{},R) & \geq \max_{0\leq \rho_2\leq 2^{C_{}}/2} \prob\left(\min_k \min_{S\in S_k} C(S\in S_k)<R\right) \nonumber \\
& \geq \max_{0\leq \rho_2\leq 2^{C_{}}/2}\max_k P_{\rm out}(S_{1,k}).
\end{align}

\end{proof}

\section{For a $N_r \times 2$ MIMO channel the events $\{C(\{1\})<R\}$ and $\{C(\{2\})<R\}$ are disjoint}
\label{sec:appb}
To prove that $\{C(\{1\})<R\}$ and $\{C(\{2\})<R\}$ are disjoint,  we will show that
\begin{align}
C(\{1\})<R \implies C(\{2\})>R,
\end{align}

We start by recalling the explicit expressions for $C(\{1\}$ and $C(\{2\}$
\begin{align}
    C(\{1\})&=2\log\left(1+\begin{bmatrix}\svv{P}_{1,1}\\\svv{P}_{2,1}\end{bmatrix}^H\begin{bmatrix}
{\rho_1} & \svv{0} \\ \svv{0} & {\rho_2}
\end{bmatrix}\begin{bmatrix}\svv{P}_{1,1}\\\svv{P}_{2,1}\end{bmatrix}\right) \nonumber \\
&=2\log\left(1+\rho_1\svv{P}_{1,1}^H\svv{P}_{1,1}+\rho_2\svv{P}_{2,1}^H\svv{P}_{2,1}\right)  \nonumber \\
&=2\log\left(1+\rho_1|\svv{P}_{1,1}|^2+\rho_2|\svv{P}_{2,1}|^2\right).
\end{align}
and
\begin{align}
    C(\{2\})&=2\log\left(1+\begin{bmatrix}\svv{P}_{1,2}\\\svv{P}_{2,2}\end{bmatrix}^H\begin{bmatrix}
{\rho_1} & \svv{0} \\ \svv{0} & {\rho_2}
\end{bmatrix}\begin{bmatrix}\svv{P}_{1,2}\\\svv{P}_{2,2}\end{bmatrix}\right) \nonumber \\
&=2\log\left(1+\rho_1\svv{P}_{1,2}^H\svv{P}_{1,2}+\rho_2\svv{P}_{2,2}^H\svv{P}_{2,2}\right) \nonumber  \\
&=2\log\left(1+\rho_1|\svv{P}_{1,2}|^2+\rho_2|\svv{P}_{2,2}|^2\right).
\end{align}
Since the precoding matrix $\svv{P}$ is a  $2 \times 2$ unitary matrix, we have
\begin{align}
|\svv{P}_{2,1}|^2=1-|\svv{P}_{1,1}|^2  \nonumber \\
|\svv{P}_{2,2}|^2=1-|\svv{P}_{1,2}|^2 \nonumber \\
|\svv{P}_{1,2}|^2=1-|\svv{P}_{1,1}|^2.
\end{align}
Hence,
\begin{align}
C(\{1\})&=2(\log\left(1+\rho_2+|\svv{P}_{1,1}|^2(\rho_1-\rho_2)\right))
\label{eq:rj1}
\end{align}
and
\begin{align}
C(\{2\})&=2(\log\left(1+\rho_2+(1-|\svv{P}_{1,1}|^2)(\rho_1-\rho_2)\right)).
\label{eq:72}
\end{align}

In order to show that  $\{C(\{1\})<R\}$ and $\{C(\{2\})<R\}$  are disjoint, we next show that $C({1})<R$ necessarily implies that $C({2})>R$, for all $0\leq R \leq C$.

To that end, assume that indeed $C(\{1\})<R$. By (\ref{eq:rj1}), this implies that
\begin{align}
1+\rho_2+|\svv{P}_{1,1}|^2(\rho_1-\rho_2))<2^{R/2},
\end{align}
or equivalently
\begin{align}
-1-\rho_2-|\svv{P}_{1,1}|^2(\rho_1-\rho_2))>-2^{R/2}.
\end{align}
It follows that
\begin{align*}
1+\rho_2+(1-|\svv{P}_{1,1}|^2)(\rho_1-\rho_2)&>2+\rho_1+\rho_2-2^{R/2}.
\end{align*}
By (\ref{eq:72}), we have established that
\begin{align*}
C(\{1\})<R \implies C(\{2\})>2\log\left(2+\rho_1+\rho_2-2^{R/2}\right).
\end{align*}
To show that $C(\{1\})<R \implies C(\{2\})>R$, it suffices therefore to show that
\begin{align}
2\log\left(2+\rho_1+\rho_2-2^{R/2})\right)\geq R,
\end{align}
or equivalently,
\begin{align}
2+\rho_1+\rho_2-2^{R/2}\geq 2^{R/2}.
\end{align}
Using (\ref{eq:C_rho1_rho2}), the latter is further equivalent to showing that
\begin{align*}
1+\frac{2^C}{1+\rho_2}+\rho_2-2^{R/2} &\geq 2^{R/2}
\end{align*}
or
\begin{align*}
(1+\rho_2)^2+{2^C} &\geq (1+\rho_2)\cdot2\cdot2^{R/2}.
\end{align*}
Finally, denoting $x=1+\rho_2$, this is equivalent to showing that the following holds.
\begin{align}
x^2-2\cdot2^{R/2}\cdot x+2^C\geq 0.
\end{align}
It can be easily verified that for all values of $\rho_2$, and for all \mbox{$0\leq R \leq C$}, this inequality indeed holds. Therefore, we conclude that $\{C(\{1\})<R\}$ and $\{C(\{2\})<R\}$ are disjoint.
\end{appendices}
\bibliographystyle{IEEEtran}
\bibliography{eladd}

\end{document}

%% file: rv_defs.tex
\DeclareMathAlphabet{\mathbsf}{OT1}{cmss}{bx}{n}
\DeclareMathAlphabet{\mathssf}{OT1}{cmss}{m}{sl}
\DeclareMathAlphabet{\mathcsf}{OT1}{cmss}{sbc}{n}

\newcommand{\svv}[1]{\mathbf{#1}}

\DeclareSymbolFont{bsfletters}{OT1}{cmss}{bx}{n}  
\DeclareSymbolFont{ssfletters}{OT1}{cmss}{m}{n}
\DeclareMathSymbol{\bsfGamma}{0}{bsfletters}{'000}
\DeclareMathSymbol{\ssfGamma}{0}{ssfletters}{'000}
\DeclareMathSymbol{\bsfDelta}{0}{bsfletters}{'001}
\DeclareMathSymbol{\ssfDelta}{0}{ssfletters}{'001}
\DeclareMathSymbol{\bsfTheta}{0}{bsfletters}{'002}
\DeclareMathSymbol{\ssfTheta}{0}{ssfletters}{'002}
\DeclareMathSymbol{\bsfLambda}{0}{bsfletters}{'003}
\DeclareMathSymbol{\ssfLambda}{0}{ssfletters}{'003}
\DeclareMathSymbol{\bsfXi}{0}{bsfletters}{'004}
\DeclareMathSymbol{\ssfXi}{0}{ssfletters}{'004}
\DeclareMathSymbol{\bsfPi}{0}{bsfletters}{'005}
\DeclareMathSymbol{\ssfPi}{0}{ssfletters}{'005}
\DeclareMathSymbol{\bsfSigma}{0}{bsfletters}{'006}
\DeclareMathSymbol{\ssfSigma}{0}{ssfletters}{'006}
\DeclareMathSymbol{\bsfUpsilon}{0}{bsfletters}{'007}
\DeclareMathSymbol{\ssfUpsilon}{0}{ssfletters}{'007}
\DeclareMathSymbol{\bsfPhi}{0}{bsfletters}{'010}
\DeclareMathSymbol{\ssfPhi}{0}{ssfletters}{'010}
\DeclareMathSymbol{\bsfPsi}{0}{bsfletters}{'011}
\DeclareMathSymbol{\ssfPsi}{0}{ssfletters}{'011}
\DeclareMathSymbol{\bsfOmega}{0}{bsfletters}{'012}
\DeclareMathSymbol{\ssfOmega}{0}{ssfletters}{'012}